\documentclass[11pt]{article}
\usepackage{fullpage}
\usepackage{url}
\usepackage{amsmath, amssymb, amsthm, color}
\usepackage{graphicx}
\usepackage{subfigure}
\usepackage{float}
\usepackage{enumerate}
\usepackage[margin=1cm]{caption}

\newtheorem{lemma}{Lemma}

\newtheorem{theorem}{Theorem}

\newcommand\mC{\mathcal{C}}

\begin{document}
\title{An Improved Algorithm for Fixed-Hub Single Allocation Problem}
\author{Dongdong Ge\thanks{Research Center for Management Science and Data Analytics, Shanghai University of Finance and Economics, Shanghai,
China 200433. Email: {\tt dongdong@gmail.com}} \and Zizhuo
Wang\thanks{Department of Industrial and Systems Engineering,
University of Minnesota, Minneapolis, MN, 55455, USA. Email: {\tt
zwang@umn.edu}} \and Lai Wei\thanks{Stephen M. Ross School of
Business, University of Michigan, Ann Arbor, MI, 48109, USA. Email:
{\tt laiwi@umich.edu}} \and Jiawei Zhang\thanks{Leonard N. Stern
School of Business, New York University, New York, NY, 10012, USA.
Email: {\tt jzhang@stern.nyu.edu}}
 }\maketitle

\maketitle

\begin{abstract}
\noindent This paper discusses the fixed-hub single allocation
problem (FHSAP). In this problem, a network consists of hub nodes
and terminal nodes. Hubs are fixed and fully connected; each
terminal node is connected to a single hub which routes all its
traffic. The goal is to minimize the cost of routing the traffic in
the network. In this paper, we propose a linear programming
(LP)-based rounding algorithm. The algorithm is based on two ideas.
First, we modify the LP relaxation formulation introduced in Ernst
and Krishnamoorthy (1996, 1999) by incorporating a set of validity
constraints. Then, after obtaining a fractional solution to the LP
relaxation, we make use of a geometric rounding algorithm to obtain
an integral solution. We show that by incorporating the validity
constraints, the strengthened LP often provides much tighter upper
bounds than the previous methods with a little more computational
effort, and the solution obtained often has a much smaller gap with
the optimal solution. We also formulate a {\it robust} version of
the FHSAP and show that it can guard against data uncertainty with
little cost.

\bigskip
\noindent {\it Key words:} hub location; network design; linear
programming; worst-case analysis
\end{abstract}

\section{Introduction}
\label{sec:introduction} Hub-and-spoke networks have been widely
used in transportation, logistics, and telecommunication systems. In
such networks, traffic is routed from numerous nodes of origin to
specific destinations through hub facilities. The use of hub
facilities allows for the replacement of direct connections between
all nodes with fewer, indirect connections. One main benefit is the
economies of scale as a result of the consolidation of flows on
relatively few arcs connecting the nodes. In the United States,
hub-and-spoke routing is practically universal. Airlines adopted it
after the industry was deregulated in $1978$. Many logistics service
providers such as UPS and Fedex also have distribution systems using the
hub-and-spoke structure.

Given its widespread use, it is of practical importance to design
efficient hub-and-spoke networks. In the literature, such problems
are often referred as the hub location problems, in which two major
questions are studied: 1) where the hubs should be located and 2)
how the traffic/flow should be routed. We refer the readers to
\cite{campbell02} for a comprehensive review of the literature on
the hub location problems.

In this paper, we focus on a sub-problem which is called the
fixed-hub single allocation problem (FHSAP). In the FHSAP, the
locations of the hubs are fixed, and the decisions are to assign
each terminal node to a unique hub. Although the FHSAP is a
sub-problem of the hub location problems, it is still of great
interest. First, in many practical situations, the locations of the
hubs are pre-determined and remain unchanged in a medium to long
term. In such cases, the hubs can be viewed as fixed and only the
assignment of the terminal nodes needs to be decided. Second, the
number of nodes that can be used as hubs are usually small, which
makes it possible to enumerate all possible locations of the hubs to
find the optimal location. Therefore, solving the fixed-hub
allocation problem efficiently would be of great help for solving
the hub location problem. Moreover, even confined to fixed hubs,
optimally assigning terminal nodes to hubs is still a challenging
task. Indeed, it is known that FHSAP is NP-Hard even for problem
with 3 hubs \cite{sohn97}. Therefore, designing efficient algorithms
to solve FHSAP is of great interest, both to researchers and
practitioners.

To address the FHSAP, several prior approaches have been proposed.
In O'Kelly \cite{okelly87}, the author proposes a quadratic integer
program to model this problem. The formulation is non-convex and
thus hard to solve. Therefore, the author proposes two heuristic methods to
solve it. Following \cite{okelly87}, several other heuristic methods are
proposed, see, e.g., Klincewicz \cite{klince91}, Campbell
\cite{campbell96} and Skorin-Kapov et al. \cite{skorin96}.

One major method to solve the FHSAP is to use a linearization model
for the quadratic integer program in \cite{okelly87}. Such
linearizations are developed in \cite{campbell94b, okelly95,
okelly96, ernst1, ernst2}. One of the earliest such linearization
model is introduced by \cite{campbell94b}, in which a natural LP
relaxation of the quadratic integer program is obtained. This LP
relaxation is quite attractive: Skorin-Kapor et al \cite{skorin96}
show that it is very tight and outputs integral solutions
automatically in $95\%$ of the instances they test. However, the
size of this LP relaxation is relatively large and thus restricts
its applications to large-scale problems. To solve this problem,
Ernst and Krishnamoorthy \cite{ernst1, ernst2} propose a further
relaxation of the model. The idea of the further relaxation is to
use combined flow variables, and the size of the further relaxed LP
is significantly smaller than that in \cite{campbell94b} and
\cite{skorin96}. However, in some situations, the further relaxed
model has a large gap with the optimal solution.

In this paper, we propose a new LP relaxation for the FHSAP. Our new
LP relaxation is based on the one proposed in \cite{ernst1, ernst2},
but we add a set of flow validity constraints to it. We show that by
adding the flow validity constraints, we can often tighten the gap
between the LP relaxation of \cite{ernst1, ernst2} and the optimal
solution, and yield integral solutions more frequently. Moreover, it
comes with reasonable computational cost. Therefore, we believe our
approach is a good balance between the LP relaxation by
\cite{campbell94b} and \cite{ernst1, ernst2}.

Besides finding a suitable LP relaxation, another important question
is how to round a fractional solution of the LP into a feasible
solution to the FHSAP. In this paper, we adopt a geometric rounding
algorithm introduced by Ge et al. \cite{ge2011}. In \cite{ge2011},
the authors propose a random geometric rounding scheme for a class
of assignment problems. They prove that this rounding technique can
be applied to the FHSAP and lead to a constant ratio approximation
algorithm under the equilateral structure. In this paper, we show
that our newly proposed LP relaxation combined with the geometric
rounding algorithm yields good solutions to the FHSAP efficiently.
It is worth noting that another dependent rounding scheme by
Kleinberg and Tardos \cite{kt2001} can also be adopted to round the
solutions.

In practical cases, the demands in the FHSAP may be unknown. To
tackle such situations, we propose a {\it robust} programming
approach for the FHSAP when the demands are only known to be within
a certain convex set. We derive a convex programming relaxation for
the {\it robust} formulation which can be solved efficiently. We
show in our numerical tests that by employing the decisions of the
{\it robust} model, we can guard against the demand uncertainty with
little cost, therefore it might be of practical interest.

The remainder of the paper is organized as follows. In Section
\ref{sec:model}, we introduce the model and the LP relaxation we
propose for the FHSAP. Then we introduce the geometric rounding
scheme in Section \ref{sec:rounding}. In Section
\ref{sec:numerical}, we perform numerical tests to show that our
proposed approach can indeed obtain better solutions to this
problem. In Section \ref{sec:robust}, we establish a {\it robust} model
for the FHSAP, and study the solution of the {\it robust} model. Then we
conclude our paper in Section \ref{sec:conclusion}.

\section{Model and Formulation}
\label{sec:model} This section defines the fixed-hub single
allocation problem, reviews and modifies previously proposed
mathematical programs. By the terminology of communication networks,
the problem is to build a two-level network consisting of {\it hubs}
and {\it terminal nodes} , see Figure \ref{fig:illustration} for an
illustration. In the FHSAP, we assume that there are $k$ fixed hubs
denoted by $\mathcal{H}=\{1,2,\ldots, k\}$ (airports, routers,
concentrators, etc.), which are transit nodes that are used to route
traffic. There are $n$ terminals nodes denoted by
$\mathcal{C}=\{1,2,\ldots, n\}$ (cities, computers, etc.) which
represent the origins and the destinations of the traffic. Here all
hubs are fully connected and each terminal node is connected to
exactly one hub.

\begin{figure}
\centering
\includegraphics[width=3.0in]{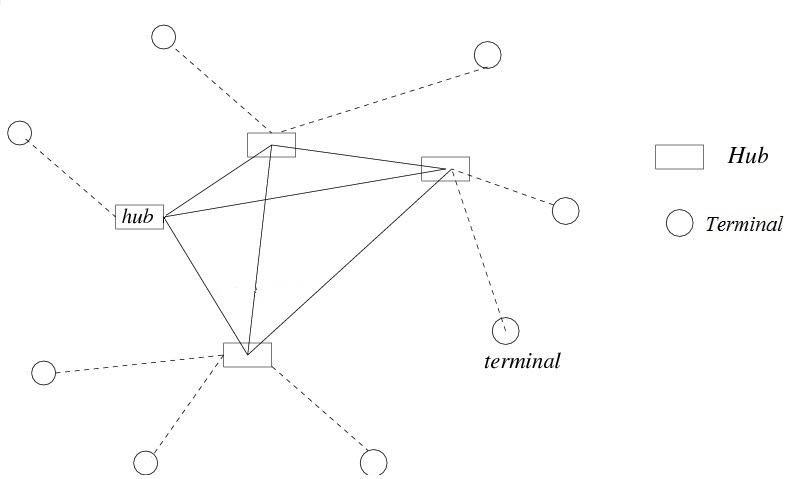}
%{illustration.jpg}
\caption{An illustration of the two-level network.}
\label{fig:illustration}
\end{figure}

In this network, there is a demand $d_{ij}$ to be routed from $i$ to
$j$, for each pair of terminal nodes $i$ and $j$. In order to route
the demands between two terminal nodes, the origin node has to
deliver all its demands to the hub it is assigned to. Then this hub
sends them to the hub the destination node is assigned to (this step
is skipped if both nodes are assigned to the same hub). Finally the
destination node gets the demands from its hub. No direct routing
between two terminal nodes is permitted. Two types of costs are
counted during the transportation, a per unit transportation cost
$c_{is}$ to transport demand from terminal node $i$ to hub $s$ and a
per unit transportation cost $c_{st}$ to transport demand from hub
$s$ to hub $t$. The problem is to assign a hub for each terminal
node such that the total transportation cost is minimized.

The first mathematical formulation to study the FHSAP is by O'Kelly
\cite{okelly87}, in which he formulates the problem as a quadratic
integer program.\footnote{In fact, his formulation is for a more
general problem, the uncapacitated single allocation $p$-hub median
problem. In this paper, we only confine our discussion to the FHSAP
and thus adapt his formulation (and later formulations) to the
FHSAP.} Define $\vec{x}=\{x_{is}:i\in\mathcal{C}, s\in\mathcal{H}
\}$ to be the assignment variables. The quadratic formulation for
the FHSAP is:

{\bf Problem FHSAP-QP}
\begin{eqnarray*}
 \mbox{minimize} & {\displaystyle \sum_{i,j\in \mathcal{C}}
d_{ij}\left( \sum_{s\in \mathcal{H}} c_{is} x_{is}+ \sum_{t\in
\mathcal{H}}
c_{jt}x_{jt} +\sum_{s,t\in \mathcal{H}} c_{st}x_{is}x_{jt} \right)} & \\
 \mbox{subject to}
&\displaystyle \sum_{s\in \mathcal{H}}x_{is}=1, & \forall i\in \mathcal{C},\\
 & x_{is}\in\left\{ 0,1\right\},  & \forall i\in\mathcal{C},
 s\in\mathcal{H}.
\end{eqnarray*}

Here we assume that all coefficients $d_{ij}, c_{is}, c_{jt},
c_{st}$ are non-negative and $c_{st}=c_{ts}$, $c_{ss}=0$, for all
$i,j\in\mathcal{C}$ and $s,t\in\mathcal{H}$. Note that the
transportation cost from cities to hubs, $ \sum_{i,j\in
\mathcal{C}}d_{ij}( \sum_{s\in \mathcal{H}}c_{is} x_{is}+\sum_{t\in
\mathcal{H}}c_{jt}x_{jt})$, is linear in $\vec{x}$. Later we call it
the {\em linear cost} and denote it by $L(\vec{x})$. Similarly, we
call the other part of the objective function the {\em inter-hub
cost} or {\em quadratic cost}, and denote it by $Q(\vec{x})$.

Campbell~\cite{campbell94b} linearized O'Kelly's model by
formulating a mixed integer linear program (MILP) as follows:

{\bf Problem MILP1}
\begin{eqnarray*}
\mbox{minimize} & {\displaystyle \sum_{i,j\in \mathcal{C}}
\sum_{s,t\in
\mathcal{H}}d_{ij}(c_{is}+c_{st}+c_{jt})X_{ijst}} & \\
\mbox{subject to}
&\displaystyle \sum_{s,t\in \mathcal{H}}X_{ijst}=1, & \forall i,j\in \mathcal{C},\\
&\displaystyle \sum_{t\in \mathcal{H}}X_{ijst}=x_{is}, & \forall i,j\in \mathcal{C}, s\in \mathcal{H},\\
&\displaystyle \sum_{s\in \mathcal{H}}X_{ijst}=x_{jt}, & \forall i,j\in \mathcal{C}, t\in \mathcal{H},\\
& X_{ijst}\geq 0, & \forall i,j\in\mathcal{C}, s,t\in\mathcal{H},\\
& x_{is}, x_{jt}\in\left\{ 0,1\right\},  & \forall i\in\mathcal{C},
s\in\mathcal{H}.
\end{eqnarray*}

Here $X_{ijst}$ is the portion of the flow from city $i$ to city $j$
via hub $s$ and $t$ sequentially. The formulation involves
$O(n^2k^2)$ nonnegative variables and $O(n^2k)$ constraints. This
formulation enables us to obtain an LP relaxation for the FHSAP by
replacing the zero-one constraints with non-negative constraints. In
the following, we refer this LP relaxation as LP1. As shown in
\cite{skorin96}, LP1 is usually very tight and often produces
integral solutions. However, the size of LP1 is relative large,
which restricts its applications to large-scale problems.

In order to reduce the solution complexity, Ernst and Krishnamoorthy
\cite{ernst1,ernst2} propose a flow formulation to obtain a further
relaxation of this problem. In this formulation, one does not need
to specify the route for a pair of terminal nodes $i$ and $j$, i.e.,
one does not need the decision variable $X_{ijst}$. Instead, one
defines $\vec{Y}=\{Y^i_{st}: i \in \mathcal{C}, s, t \in
\mathcal{H}, s\neq t\}$ where $Y^i_{st}$ is the total amount of the
flow originated from city $i$ and routed from hub $s$ to a different
hub $t$. Then the FHSAP can be further relaxed to:

{\bf Problem MILP2}
\begin{eqnarray}
\mbox{minimize} & {\displaystyle \sum_{i\in \mathcal{C}} \sum_{s\in
\mathcal{H}}c_{is}(O_i+D_i)x_{is}+\sum_{i \in
\mathcal{C}}\sum_{s,t \in \mathcal{H} :s\neq t}c_{st}Y^i_{st}} & \nonumber\\
\mbox{subject to}
&\displaystyle \sum_{s\in \mathcal{H}}x_{is}=1, & \forall i\in \mathcal{C},\nonumber\\
&\displaystyle \sum_{t\in \mathcal{H}:t\neq s}Y^i_{st}-\sum_{t\in
\mathcal{H}:t\neq s}Y^i_{ts}=O_i x_{is}-\sum_{j \in
\mathcal{C}}d_{ij}x_{js},
& \forall i\in \mathcal{C}, s\in \mathcal{H},\nonumber\\
& x_{is}\in\left\{ 0,1\right\},  & \forall i\in\mathcal{C},
s\in\mathcal{H},\nonumber\\
& Y^i_{st} \ge 0, & i \in \mathcal{C}, s, t \in \mathcal{H}, s\neq
t. \label{fhsap-milp2}
\end{eqnarray}
where $O_i=\sum_{j\in \mathcal{C}}d_{ij}$ and $D_i=\sum_{j \in
\mathcal{C}}d_{ji}$ denote the total demands from and to $i$
respectively. Note that this modified formulation involves only
$O(nk^2)$ nonnegative variables and $O(nk)$ linear constraints,
which decreases from that of MILP1 by a factor of $n$. We can then
obtain an LP relaxation from MILP2, which we denote by LP2.

To see that MILP2 is indeed a further relaxation of the problem,
note that any feasible assignment $\vec{x}$ to the FHSAP with the
flow vector $\vec{Y}$ is always a feasible solution to MILP2 with
the objective value equal to the transportation cost. Since MILP2
reduces the formulation size by $n$, its LP relaxation is also
easier to solve. However, despite that it is proved that MILP2 is an
exact formulation when all the costs in the system are equal, in
general, there might be a positive gap between the optimal value of
MILP2 and the true optimal solution. And in our numerical tests, we
find that the gap sometimes is quite large. Therefore, it is useful
to find an improved formulation of MILP2 without adding too much
complexity.

In the following, we propose a stronger formulation than MILP2. The
main idea is to add a set of validity constraints based on the
following observation.

\begin{lemma}\label{lemma-flow-cons}
Let $\vec{x}$ and $\vec{Y}$ be defined as in {\em MILP2}. For any $i
\in \mathcal{C}$ and $s\in \mathcal{H}$, we have
\begin{equation}\label{flow-cons}
\sum_{t\in \mathcal{H} :t\neq s }Y^i_{st}+\sum_{t\in \mathcal{H}
:t\neq s}Y^i_{ts}= \sum_{j \in \mathcal{C}}d_{ij}|x_{is}-x_{js}|.
\end{equation}
\end{lemma}

\begin{proof}
We verify equation (\ref{flow-cons}) in two cases.

\begin{enumerate}
\item If $x_{is}=0$, then $$\sum_{t\in \mathcal{H} :t\neq
s}Y^i_{st}+\sum_{t\in \mathcal{H} :t\neq s}Y^i_{ts}=\sum_{t\in
\mathcal{H}: t \neq s}Y^i_{ts}=\sum_{j \in \mathcal{C}}d_{ij} x_{js}
=\sum_{j \in \mathcal{C}}d_{ij}|x_{is}-x_{js}|.$$

\item If $x_{is}=1$, then
$$\sum_{t\in \mathcal{H} :t\neq s}Y^i_{st}+\sum_{t\in
\mathcal{H} :t\neq s}Y^i_{ts}=\sum_{t\in \mathcal{H}: t \neq
s}Y^i_{st} =\sum_{j \in \mathcal{C}: x_{js}=0}d_{ij}=\sum_{j \in
\mathcal{C}}d_{ij}(1-x_{js})=\sum_{j \in
\mathcal{C}}d_{ij}|x_{is}-x_{js}|.$$
\end{enumerate}

Therefore, equation (\ref{flow-cons}) holds in both cases.
\end{proof}

Based on Lemma \ref{lemma-flow-cons}, we obtain a strengthened
formulation of (\ref{fhsap-milp2}) with additional constraints
\begin{eqnarray*}
\sum_{t\in \mathcal{H} :t\neq s }Y^i_{st}+\sum_{t\in \mathcal{H}
:t\neq s}Y^i_{ts}= \sum_{j \in \mathcal{C}}d_{ij}y_{ijs}\\
x_{is} - x_{js} \le y_{ijs}\\
x_{js} - x_{is} \le y_{ijs}.
\end{eqnarray*}
We call this problem MILP2' and its LP relaxation LP2'. Note that
LP2' has both $O(n^2k+nk^2)$ variables and constraints. We further
reduce the number of additional constraints by summing up the
validity constraints. We get our final formulation as follows:

{\bf Problem MILP3}
\begin{eqnarray}
\mbox{minimize} & {\displaystyle \sum_{i\in \mathcal{C}} \sum_{s\in
\mathcal{H}}c_{is}(O_i+D_i)x_{is}+\sum_{i \in
\mathcal{C}}\sum_{s,t \in \mathcal{H} :s\neq t}c_{st}Y^i_{st}} & \nonumber\\
\mbox{subject to}
&\displaystyle \sum_{s\in \mathcal{H}}x_{is}=1, & \forall i\in \mathcal{C},\nonumber\\
&\displaystyle \sum_{t\in \mathcal{H} :t\neq s}Y^i_{st}-\sum_{t\in
\mathcal{H}: t\neq s}Y^i_{ts}=O_i x_{is}-\sum_{j \in
\mathcal{C}}d_{ij}x_{js}, &
\forall i\in \mathcal{C}, s\in \mathcal{H},\nonumber\\
&\displaystyle 2\sum_{i \in \mathcal{C}}\sum_{s,t\in \mathcal{H}
:s\neq t }Y^i_{st}= \sum_{i,j \in \mathcal{C}}\sum_{s \in
\mathcal{H}}d_{ij}y_{ijs}, & \nonumber\\
& x_{is}-x_{js}\leq y_{ijs}, & \forall i,j\in \mathcal{C}, s\in
\mathcal{H},\nonumber\\
& x_{js}-x_{is}\leq y_{ijs}, & \forall i,j\in \mathcal{C}, s\in
\mathcal{H},\nonumber\\
& x_{is}\in\{0,1\} \nonumber\\
 & Y^i_{st}, x_{is}, y_{ijs} \ge 0, & \forall i \in \mathcal{C}, s, t \in
\mathcal{H}, s\neq t. \label{flow-constraints-in-LP}
\end{eqnarray}

We call the LP relaxation of MILP3 by LP3. The number of variables
and constraints in MILP3 and LP3 are both $O(n^2 k+nk^2)$. Although
it doesn't reduce the size of LP2' significantly, computational
results indicate that LP3 can be solved much more efficiently, yet
the results are usually quite good.

The above formulations can serve two purposes. First, it provides a
tighter lower bound for the FHSAP than LP2. Second, it provides a
new way to solve the FHSAP using LP relaxations. In the next
section, we show how to obtain an integral solution from the
fractional solution solved from the LP relaxations. In Section
\ref{sec:numerical}, we perform numerical tests to show the
performance of our proposed approach.

\section{Rounding Procedure: A Geometric Rounding Algorithm}
\label{sec:rounding}

Note that in the above formulations, a solution to the FHSAP can be
completely defined by the assignment variables $\{x_{ik}\}$.
Therefore, after solving an LP relaxation (LP1, LP2 or LP3), we only
need to focus on rounding the fractional assignment variables to
binary integers. Note that in the three relaxations presented above
(LP1, LP2 or LP3), we all have the constraints that $\sum_s x_{is} =
1$. Therefore, for a terminal node $i$, any optimal solution
$x_i=(x_{i1}, \ldots, x_{ik})$ of the LP relaxation must fall on the
standard $k-1$ dimensional simplex:
\begin{eqnarray*}
\Delta_k = \left\{ w\in \mathbb{R}^k \left| w\geq 0, \sum_{i=1}^k
w_i=1 \right.\right\}.
\end{eqnarray*}
A fractional assignment vector on node $i$ corresponds to a
non-vertex point of $\Delta_{k}$. Our goal is to round any
fractional solution to a vertex point of $\Delta_{k}$, which is of
the form:
\begin{eqnarray*}
\left\{w\in \mathbb{R}^k \left| w_i \in \{0,1\}, \sum_{i=1}^k
w_i=1\right.\right\}.
\end{eqnarray*}
It is clear that $\Delta_k$ has exactly $k$ vertices. In the
following, we denote the vertices of $\Delta_k$ by $v_1, v_2,
\cdots, v_k$, where the $i$th coordinate of $v_i$ is $1$.

Before presenting the rounding procedure, we define some notations
of the geometry of the problem. For a point $x\in \Delta_{k}$,
connect $x$ with all vertices $v_1,\ldots,v_k$ of $\Delta_{k}$.
Denote the polyhedron which exactly has vertices $\{x,v_1,\ldots,
v_{i-1},v_{i+1},\ldots,v_k\}$ by $A_{x,i}$. Thus the simplex
$\Delta_{k}$ can be partitioned into $k$ polyhedrons $A_{x,1},
\ldots, A_{x,k}$, and the interiors of any distinct pair of these
$k$ polyhedrons do not intersect. %Denote the volume of $A_{x,i}$ by
%$V_{x,i}$, and the volume of $\Delta_k$ by $V_k$.

\begin{figure}
\centering
\includegraphics[width=2.5in]{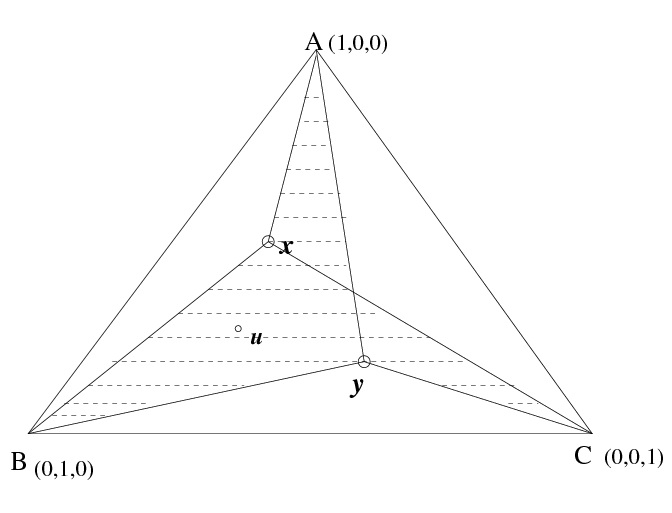}
\caption{By
the geometric rounding method, $\hat{x}=(1,0,0)$, $\hat{y}=(0,0,1)$
as the graph indicates.} \label{geo-rounding}
\end{figure}
We are now ready to present our randomized rounding algorithm. Note
that this rounding procedure is applicable to problems other than
FHSAP, as long as the feasible set of the problems is the set of
vertices of a simplex.

\noindent{\bf Geometric Rounding Algorithm (FHSAP-GRA)}:
\begin{enumerate}
\item Solve an LP relaxation of the FHSAP (LP1, LP2 or LP3). Denote the optimal solution by $\vec{x}^*$.

\item  Generate a random vector $u$, which follows a uniform
distribution on $\Delta_{k}$.

\item For each $x^*_i=(x^*_{i1},\dots, x^*_{ik})$, if $u$ falls
into $A_{x^*_i,s}$, let $\hat{x}_{is}=1$; other components
$\hat{x}_{it}=0$.

\item Output $\hat{x}$.
\end{enumerate}

An illustration of this procedure is shown in Figure
\ref{geo-rounding}. Next we discuss some theoretical properties of
this rounding technique. We first consider the uniform cost case, in
which the inter-hub costs $c_{st}$ are all equal. In this case, an
important observation is that for any integral solution $\{x_{is}\}$
\begin{eqnarray*} 2\sum_{s,t\in\mathcal{H}}
\hat{x}_{is}\hat{x}_{jt}= 2\sum_{s,t\in\mathcal{H}: s \neq t}
\hat{x}_{is}\hat{x}_{jt}=
2\left(1-\sum_{s\in\mathcal{H}}\hat{x}_{is}\hat{x}_{js}\right)
=2\sum_{s\in\mathcal{H}}|\hat{x}_{is}-\hat{x}_{js}|.
\end{eqnarray*}

The above observation is made by Kleinberg and Tardos \cite{kt2001}
and some other literature thereafter. Furthermore, by exploring the
structure of the geometric rounding, Ge et al. \cite{ge2011} proved the
following two properties of this rounding technique.

\begin{theorem}\label{thm:equil-linearcost}
For any given $i\in \mathcal{C},l\in \mathcal{H}$,
$\mathbb{E}[\hat{x}_{il}]=x^*_{il}$.
\end{theorem}

For any $x$ and $y$, define $d(x,y):=\sum_s |x_{s}-y_{s}|$. Then
\begin{theorem}\label{thm:dis}
For any $x,y\in \Delta_{k}$, if we randomly round $x$ and $y$ to
vertices $\hat{x}$ and $\hat{y}$ in $\Delta_{k}$ by the procedure in
{\em FHSAP-GRA}, then
$$\mathbb{E}[d(\hat{x}, \hat{y})]\leq 2d(x,y).$$
\end{theorem}

Furthermore, by Chekuri et al \cite{chekuri}, if $(\vec{x}^*,
\vec{X}^*)$ is an optimal solution to LP1, then
\begin{eqnarray}\label{relation}
\sum_{s\in \mathcal{H}} |x_{is}^* - x_{js}^*| \le \sum_{s\in
\mathcal{H}}\sum_{t\in\mathcal{H}, t\neq s}\left(X_{ijst}^* +
X_{ijts}^*\right).
\end{eqnarray}

Combining Theorem \ref{thm:equil-linearcost}, \ref{thm:dis},
inequality (\ref{relation}) and the fact that LP1 provides a lower
bound of FHSAP-QP, we have the following theorem for our geometric
rounding algorithm:
\begin{theorem}\label{thm:bound}
Assume all $c_{st}$'s are equal. Let $\vec{x}^*$ be the optimal
solution of LP1. And $\hat{x}$ be the integral solution obtained
from FHSAP-GRA. Then $\hat{x}$ is a 2-approximation for the original
problem. That is
\begin{eqnarray*}
L(\hat{x}) + Q(\hat{x}) \le 2 {\mbox{OPT}}
\end{eqnarray*}
where {\mbox{OPT}} is the optimal value of FHSAP-QP.
\end{theorem}

In a more general case when the inter-hub costs are not equal, we
have the following extension of Theorem \ref{thm:bound}.

\begin{theorem}\label{thm:extension}
Let $L = \max\{c_{st}\colon s,t\in \mathcal{H}, s\neq t \}$, $ \l =
\min\{c_{st}\colon s,t\in \mathcal{H}, s\neq t\}$ and $r = L/l$.
Then the algorithm FHSAP-GRA using the LP relaxation LP1 has a
performance guarantee of $2r$.
\end{theorem}

\begin{proof}
Given an instance $I$ of the FHSAP, we build another instance
denoted by $I_L$, in which all of the inter-hub costs are set to
$L$. Obviously, for any allocation vector $\vec{x}$, the cost in
$I_L$ is less than $r$ times the cost in $I$. And by using the algorithm
FHSAP-GRA, we can obtain a 2-approximation for $I_L$, which directly
translates into a $2r$-approximation for $I$.
\end{proof}

It is tempting to extend the above results to the case when LP3 is
used. However, since LP3 is a further relaxation of LP1 and the
structure of the objective is not the same, we are not able to prove
the same bound when we use GRA combined with LP3. However, as we
will show in our numerical experiments, using LP3 combined with the
geometric rounding procedure produces good solutions.

\section{Computational Results}
\label{sec:numerical}

In this section, we implement our algorithm (FHSAP-GRA) and report
its performance. We test the algorithm on both randomly generated
instances (Table \ref{table:randomdata1}) and a benchmark problem
(Table \ref{table:benchmark}). All linear programs in the
experiments are solved by CPLEX version $9.0$ at a workstation with
3GHz CPU and 8GB memory.

In Table \ref{table:randomdata1}, we consider three setups of the
problem: $n = 50$, $k=5$; $n=100$, $k=10$ and $n=200$, $k=10$. In
each of the setup, demands between any two cities are generated from
uniform distributions $U[0,100]$ and hub to city costs are generated
from $U[1,11]$. Then we choose different distributions for the
inter-hub costs to conduct our tests, which are shown in the second
column. We try to solve all the three LP relaxation problems (LP1,
LP2 and LP3) introduced in Section \ref{sec:model} and apply the
geometric rounding algorithm to the solutions we obtain. The results
are shown in the three columns GRA-LP1, GRA-LP2 and GRA-LP3
respectively. Within each of the three sets of experiments, the CPU
columns show the time (in seconds) our program takes to solve each
LP relaxation (we find that the time to perform the rounding
procedure is negligible, therefore we only report the time to solve
the LPs in our test results). Gap-LP1 columns show the gap between
the cost of our obtained integral solution (for each fractional
solution we obtain, we do the random rounding 5000 times and pick
the best results) and the optimal value of LP1. More precisely, if
we denote the optimal value of LPi by $v_i$ and the value of an
integral solution by algorithm GRA-LPi by $w_i$, then GAP-LP1 is
$(w_i/v_1 - 1)\times 100\%$. Similarly, Gap-LP3 columns show the gap
between the cost of our obtained integral solution and the optimal
value of LP3. That is, GAP-LP3 is $(w_i/v_3 - 1)\times 100\%$. Since
all the LP relaxations are lower bounds of the optimum of the
original problem, these two columns are upper bounds of the
performance gaps between the obtained solution and the true optimal
allocation. Note that when the problem size is large, e.g., $n=200$
and $k=10$, we are not able to solve LP1, we use N/A to denote such
cases.

In Table \ref{table:randomdata1}, we can see that there are several
features of our proposed algorithm (GRA-LP3). First, although
solving LP3 is not as efficient as solving LP2, it is still mostly
tractable while the solution time of LP1 increases very fast and
soon becomes intractable. On the other hand, the solution provided
by GRA-LP3 could provide significant improvement over the solution
that is obtained by using LP2. In the $15$ tests we presented, there
are $9$ cases in which GRA-LP3 could produce the exact optimal
solution, while there are only $3$ if one uses GRA-LP2. Therefore,
we can conclude that GRA-LP3 could deliver higher-quality solutions
than GRA-LP2 within reasonable amount of time.

\begin{table}
\begin{center}
\small
\begin{tabular}{||cc||cc||ccc||ccc||}
\hline\hline

  &  & \multicolumn{ 2}{c||}{GRA-LP1} &       \multicolumn{3}{c||}{GRA-LP2} &  \multicolumn{ 3}{c||}{    GRA-LP3}
  \\
\cline{3-4}\cline{5-7}\cline{8-10}
 \raisebox{1.5ex}[0cm][0cm]{$n$ and $k$} &   \raisebox{1.5ex}[0cm][0cm]{$c_{st}$} &        CPU &       Gap-LP1 &     CPU  &       Gap-LP1 &       Gap-LP3 &       CPU &        Gap-LP1 &       Gap-LP3  \\
\hline
$n = 50$  & U[0,20]  & 3.30 &     0.00\% &       0.04 &     4.24\% &    12.19\% &      3.5 &     3.55\% &    11.45\%  \\
  $k =5$  & U[4,20]  & 3.08 &     0.00\% &       0.04 &     1.83\% &     1.83\% &     1.58 &     0.00\% &     0.00\%  \\
          & U[14,20] & 2.55 &     0.00\% &       0.04 &     4.47\% &     4.47\% &       2.2 &     0.00\% &     0.00\%  \\
          & 10       & 3.1  &     0.00\% &       0.04 &     9.25\% &     9.25\% &       1.36&     0.00\% &     0.00\%  \\
          & 20       & 2.04 &     0.00\% &       0.04 &     0.00\% &     0.00\% &      2.14 &     0.00\% &     0.00\%  \\
\hline\hline
$n = 100$ &  U[0,20] &15249 &     0.00\% &       0.85 &    10.95\% &    51.17\% &         148 &    10.95\% &    51.17\% \\
$k=10$    &  U[4,20] &16851 &     0.00\% &       3.12 &     2.76\% &    15.07\% &         329 &     2.30\% &    14.55\% \\
          & U[14,20] &15439 &     0.00\% &       3.22 &     5.86\% &     7.47\% &         322 &     0.92\% &     2.45\% \\
          & 10       &10103 &     0.00\% &       1.08 &    9.25\%  &     9.25\% &         230 &     0.00\% &     0.00\% \\
          & 20       &13780 &     0.00\% &       4.07 &     0.00\% &     0.00\% &         310 &     0.00\% &     0.00\% \\
\hline\hline
$n=200$   & U[0,20]  & N/A  &    N/A     &       22.5 &     N/A    &    33.11\% &        2549 &    N/A     &    33.11\% \\
$k=10$    & U[4,20]  & N/A  &    N/A     &       23.1 &     N/A    &    11.88\% &        1750 &    N/A     &    11.88\% \\
          & U[14,20] & N/A  &    N/A     &       27.3 &     N/A    &     0.72\% &        3311 &    N/A     &    0.00\%  \\
          & 10       & N/A  &    N/A     &       20.2 &     N/A    &     5.04\% &        1981 &    N/A     &    0.00\%  \\
          & 20       & N/A  &    N/A     &       32.7 &     N/A    &     0.00\% &        3278 &    N/A     &    0.00\%  \\
          \hline\hline
\end{tabular}\caption{Computational
results}\label{table:randomdata1}
\end{center}
\end{table}

Next we test our algorithm using a benchmark problem set AP ({\em
Australia Post}), which was collected from a real postal delivery
network in Australia, see \cite{ernst1}. In \cite{ernst1} and
\cite{ernst2}, Ernst and Krishnamoorthy solve the {\em p}-hub
location problems for the {\em AP} data set, and we test our
algorithms using the hubs their solutions specified. In particular,
some of the hub-to-city cost coefficients are non-symmetric in the
{\em AP} data set. In our experiment, we make adjustment to it
accordingly by specifying in-flow and out-flow coefficients
separately for each $x_{is}$. The results of our tests are shown in
Table \ref{table:benchmark}.

\begin{table}
\begin{center}
\small
\begin{tabular}{||cc|c|cccc||} \hline\hline
           &            &            & \multicolumn{ 4}{c||}{              GRA-LP3}        \\
\cline{4-7}
 \raisebox{1.5ex}[0cm][0cm]{$n$} & \raisebox{1.5ex}[0cm][0cm]{$k$}&\raisebox{1.5ex}[0cm][0cm]{Optimal} &        LP3 &       GRA3 &        CPU &    Gap1  \\
 \hline\hline
        50 &          5 &     132367 &     132122 &     132372 &       6.94 &   0.004\%      \\
        50 &          4 &     143378 &     143200 &     143378 &       4.04 &   0.000\%      \\
        50 &          3 &     158570 &     158473 &     158570 &       1.92 &   0.000\%      \\
        40 &          5 &     134265 &     133938 &     134265 &       2.17 &   0.000\%     \\
        40 &          4 &     143969 &     143924 &     143969 &       1.16 &  0.000\%          \\
        40 &          3 &     158831 &     158831 &     158831 &       0.60 &   0.000\%        \\
        25 &          5 &     123574 &     123574 &     123574 &       0.23 &   0.000\%        \\
        25 &          4 &     139197 &     138727 &     139197 &       0.17 &   0.000\%         \\
        25 &          3 &     155256 &     155139 &     155256 &       0.09 &   0.000\%         \\
        20 &          5 &     123130 &     122333 &     123130 &       0.11 &   0.000\%        \\
        20 &          4 &     135625 &     134833 &     135625 &       0.08 &   0.000\%        \\
        20 &          3 &     151533 &     151515 &     151533 &       0.05 &   0.000\%        \\
        10 &          5 &      91105 &      89962 &      91105 &       0.02 &    0.000\%    \\
        10 &          4 &     112396 &     111605 &     112396 &       0.01 &   0.000\%           \\
        10 &          3 &     136008 &     135938 &     136008 &       0.01 &   0.000\%          \\
\hline\hline
\end{tabular}
\caption{AP benchmark problems.}\label{table:benchmark}
\end{center}
\end{table}

In Table \ref{table:benchmark}, we test 15 {\em AP} benchmark
problems. Since solving {\em FHSAP-LP1} already produces optimal
integral assignments for all 15 problems in less than 120 seconds,
we omit it in the table. In our results, {\em GRA-LP3} obtains
optimal assignments for 14 out of the 15 test problems, and the cost
is only $0.004\%$ higher than the optimal cost for the remaining
one. Meanwhile. the time it spends to compute the solution is much
less time than that of {\em GRA-LP1}. Therefore, we conclude that
our approach is quite reliable and efficient in solving real
problems.

\section{Robust FHSAP}
\label{sec:robust}

In previous sections, we studied the fixed-hub single allocation
problem with deterministic demand. In practice, the decision maker
may not have accurate demand information. In such cases, it is of
great interest for the decision maker to have a {\it robust} policy,
which protects him from any realization of demand. In this section,
we propose a {\it robust} formulation for the FHSAP and provide an
efficient algorithm for it.

We adopt the notations used in Section \ref{sec:model}. However,
instead of knowing the pairwise demand $\vec{d} = \{d_{ij}\}$
exactly, we only know that they are in the following set:
\begin{eqnarray}\label{demandset}
{\mathcal D} = \{\vec{d}\mbox{ } : \mbox{ } ||\Sigma^{-1} (\vec{d} -
\vec{u})||_p \le Q\}.
\end{eqnarray}
Here $\vec{u} = \{u_{ij}\}$ is the nominal demand, $\Sigma =
\mbox{diag}\{\sigma_{ij}\}$ is a weight matrix and $||\cdot||_p$ is
the $p$-norm ($p\ge 1$) of a vector defined by
\begin{eqnarray}\label{pnorm}
||x||_p = \left(\sum_{i=1}^n x_i^p\right)^{1/p}.
\end{eqnarray}
The right hand side $Q$ in (\ref{demandset}) is the ``budget'' of
robustness, indicating one's uncertainty level about the input data.
Such an uncertainty set is quite common in the {\it robust} optimization
literature with most common choices of $p$ to be $1$, $2$ or
$\infty$. For a comprehensive review of the {\it robust} optimization
literature, we refer the readers to Ben-Tal et al. \cite{bental}.

Now we consider the {\it robust} formulation for FHSAP. We start
from FHSAP-QP. In the {\it robust} formulation, we aim to minimize
the worst-case cost for any demand realization that is in the set
$\mathcal{D}$. Therefore, the {\it robust} formulation for FHSAP-QP
can be written as:
\begin{equation}\label{Robust-FHSAP}
\begin{array}{rcl}
   \mbox{minimize}_{x_{ij}\in \{0,1\}} &  \mbox{maximize}_{d_{ij}}  &   \sum_{i,j\in \mC}d_{i,j}\left[\sum_{s\in {\mathcal{H}}}c_{is}x_{is} + \sum_{t\in{\mathcal {H}}} c_{jt}x_{jt}+\sum_{s,t\in {\mathcal H}} c_{st}x_{is}x_{jt}\right] \\
                                       &  \mbox{subject to}         & ||\Sigma^{-1} (\vec{d} -  \vec{u})||_p \le
Q.
\end{array}
\end{equation}

One feature of this {\it robust} formulation is that given a set of
$\vec{x}$, the inside maximization problem has an explicit optimal
solution. Define
\begin{eqnarray*}
f_{ij} = \sum_{s\in {\mathcal{H}}}c_{is}x_{is} + \sum_{t\in{\mathcal
{H}}} c_{jt}x_{jt}+\sum_{s,t\in {\mathcal H}} c_{st}x_{is}x_{jt},
\end{eqnarray*}
the inside problem can be written as
\begin{eqnarray}
\mbox{maximize}_{\vec{d}} & \vec{f}^T \vec{d} \nonumber\\
\mbox{subject to} & ||\Sigma^{-1} (\vec{d} -  \vec{u})||_p \le Q.
\label{inside}
\end{eqnarray}
By using standard Lagrangian method, one can obtain the optimal
value to (\ref{inside}) as:
\begin{eqnarray*}
\vec{f}^T\vec{u} +  Q ||\Sigma \vec{f}||_q
\end{eqnarray*}
where $q = p/(p-1)$. Therefore, the {\it robust} counterpart of
(\ref{inside}) can be written as:

\begin{eqnarray}\label{robust_conterpart}
\mbox{minimize}_{x_{is}}  & \sum_{i,j\in {\mathcal C}}
u_{ij}\left[\sum_{s\in {\mathcal{H}}}c_{is}x_{is} +
\sum_{t\in{\mathcal {H}}} c_{jt}x_{jt}+\sum_{s,t\in {\mathcal H}}
c_{st}x_{is}x_{jt}\right]\nonumber\\
& + Q \left(\sum_{i,j} \left(\sigma_{ij} (\sum_{s\in
{\mathcal{H}}}c_{is}x_{is} + \sum_{t\in{\mathcal
{H}}}c_{jt}x_{jt}+\sum_{s,t\in {\mathcal H}}
c_{st}x_{is}x_{jt})\right)^q\right)^{1/q}\nonumber\\
\mbox{subject to} & x_{is} \in \{0,1\}, \forall i, s.\nonumber
\end{eqnarray}
Now if the inter-hub costs are the same (w.l.o.g., all equals to
one), we can write $\sum_{s,t\in {\mathcal H}} c_{st}x_{is}x_{jt}$
as $\sum_{r\in {\mathcal H}} |x_{ir} - x_{jr}|$. Further relaxing
the binary constraints on $x_{is}$, we obtain a convex optimization
problem:
\begin{eqnarray}\label{robust_counterpart2}
\mbox{minimize}_{x_{is}, Z_{ij}} & \sum_{i,j} u_{ij}Z_{ij} + Qt
\nonumber\\
\mbox{subject to} & Z_{ij} \ge \sum_{s\in {\mathcal{H}}}c_{is}x_{is}
+ \sum_{t\in{\mathcal {H}}}c_{jt}x_{jt}+ \sum_{r\in {\mathcal H}}
|x_{ir} - x_{jr}| \nonumber\\
& t \ge ||\Sigma Z||_q\\
& \sum_{s\in {\mathcal H}} x_{is} = 1, \quad\quad\forall i\nonumber\\
& 0\le x_{is} \le 1, \quad\quad\forall i, s. \nonumber
\end{eqnarray}

In general, the optimal solution in (\ref{robust_counterpart2}) is
fractional. However, as we will show in the simulation results, in
most of our test problems, an integral solution is automatically
obtained. And for the cases in which a fractional solution exists,
we can again apply the geometric rounding technique introduced in
Section \ref{sec:rounding} to obtain an integral solution.

\subsection{Numerical Experiments}
\label{subsec:numerical_robust} In this section, we perform
numerical tests using the {\it robust} approach we proposed above. We show
that the {\it robust} approach can indeed guard against data uncertainty
without compromising much the average performance (both in terms of
computational efficiency and solution quality).

In the following we use the {\it robust} approach
(\ref{Robust-FHSAP}) with $p=2$. In such cases, the {\it robust}
counterpart (\ref{robust_counterpart2}) will be a second-order cone
program (SOCP). In the tests, we study different experimental
setups. For each setup, we consider the following two approaches:
\begin{enumerate}
\item Simply use the nominal demand to obtain the solution. We use
the GRA-LP3 approach discussed in Section \ref{sec:rounding}. The
integral solution obtained (after rounding) is denoted by
$\tilde{x}$.
\item We use the {\it robust} approach with a predetermined budget of
robustness $Q$. We then solve the program
(\ref{robust_counterpart2}) to obtain the optimal solution (combined
with the same geometric rounding procedure if the solution is
fractional). We denote this solution by $\hat{x}$.
\end{enumerate}
To evaluate the solutions, we use two performance measures. First,
we evaluate the solutions at the nominal demand. We write
$\tilde{F}(x)$ to denote the cost of using $x$ at the nominal demand
and define $Gap1$ as $\frac{\tilde{F}(\hat{x}) -
\tilde{F}(\tilde{x})}{\tilde{F}(\tilde{x})}\time 100\%$ to measure
the percentage gap between the costs in nominal cases. Second, we
compute the worst-case costs of our solutions for the demand in the
{\it robust} set we defined. That is, we compute the objective
values of (\ref{robust_counterpart2}). We use $\hat{F}(x)$ to denote
the worst-case cost of using $x$ and define $Gap2$ as
$\frac{\hat{F}(\tilde{x}) -
\hat{F}(\hat{x})}{\hat{F}(\hat{x})}\times 100\%$.

In the experiment, we test the 15 AP benchmark problems. We use the
nominal demand of AP data set as the parameter $u$ which could be
interpreted as the mean of the uncertain demand. The parameter
$\sigma$ is generated from a standard lognormal distribution, and
then multiplied by $100$. The parameter $Q$ varies with the
magnitude of demands in different benchmark cases. The results are
shown in Table \ref{table:robust}.

\begin{table}
\begin{center}
\begin{tabular}{||ccc|cc|ccc|ccc||}
  \hline\hline
  $n$ & $k$ & $Q$ & Time1 & Time2 & $\tilde{F}(\tilde{x})$ & $\tilde{F}(\hat{x})$ & $Gap1$ &
  $\hat{F}(\tilde{x})$ & $\hat{F}(\hat{x})$ & $Gap2$ \\
\hline
50    & 5     & 100   & 9.51  & 1378.08 & 164806 & 165833 & 0.62\% & 5951208 & 5741017 & 3.66\% \\
    50    & 4     & 100   & 4.87  & 810.92 & 153588 & 153802 & 0.14\% & 5388655 & 5373987 & 0.27\% \\
    50    & 3     & 100   & 2.81  & 750.52 & 163094 & 164633 & 0.94\% & 5393514 & 5364942 & 0.53\% \\
    40    & 5     & 100   & 4.52  & 317.26 & 159470 & 159470 & 0.00\% & 4394564 & 4394587 & 0.00\% \\
    40    & 4     & 100   & 3.06  & 288.05 & 155004 & 155004 & 0.00\% & 4053568 & 4053573 & 0.00\% \\
    40    & 3     & 100   & 1.97  & 297.74 & 164604 & 165257 & 0.40\% & 5037397 & 4997983 & 0.79\% \\
    25    & 5     & 200   & 1.72  & 50.18 & 154884 & 155293 & 0.26\% & 4028791 & 3958828 & 1.77\% \\
    25    & 4     & 200   & 1.42  & 41.53 & 160143 & 161396 & 0.78\% & 3795295 & 3711129 & 2.27\% \\
    25    & 3     & 200   & 1.63  & 59.45 & 159383 & 160119 & 0.46\% & 4011292 & 3964292 & 1.19\% \\
    20    & 5     & 400   & 1.31  & 18.43 & 151508 & 154375 & 1.89\% & 4242145 & 4230509 & 0.28\% \\
    20    & 4     & 400   & 0.91  & 15.02 & 155042 & 159317 & 2.76\% & 4314851 & 4232389 & 1.95\% \\
    20    & 3     & 400   & 1.28  & 12.06 & 153353 & 155595 & 1.46\% & 7627887 & 6676250 & 14.25\% \\
    10    & 5     & 600   & 2.31  & 4.51  & 122539 & 126367 & 3.12\% & 1669913 & 1652178 & 1.07\% \\
    10    & 4     & 600   & 1.50  & 4.35  & 121391 & 121524 & 0.11\% & 1805759 & 1718371 & 5.09\% \\
    10    & 3     & 600   & 0.77  & 3.35  & 137723 & 139034 & 0.95\% & 2685820 & 2481755 & 8.22\% \\
  \hline\hline
\end{tabular}\caption{Comparison between regular approach and robust approach}\label{table:robust}
\end{center}
\end{table}

In Table \ref{table:robust}, Time1 is the computational time of the
regular approach and Time2 is the computational time of the it
robust approach. First we can observe that although Time2 is larger
that Time1, it is still tractable. Regarding the performances of the
two approaches, we can see that Gap1 is smaller than $1\%$ in 11
cases out of 15 and is less than $Gap2$ in 12 cases while Gap2 is
larger than $1\%$ in 9 cases in 15 and is almost $15\%$ in the
$20$-city-$2$-hub case. Therefore, the additional cost for the
solution $\hat{x}$ that one has to pay at the nominal demand is much
smaller than the potential additional costs one needs to pay if one
use $\tilde{x}$ but the demands turn out to be adverse. Or in other
words, the benefit of the robust approach outweighs the cost.

\section{Conclusion}
\label{sec:conclusion}

In this paper, we studied the fixed-hub single allocation problem.
We made two contributions. First, we proposed a new solution
approach for this problem by establishing a new LP relaxation
formulation. This new LP relaxation lies in between two known
relaxations in the literature, and by showing numerical results of
this relaxation, we show that it has a good balance between
computational complexity and the solution quality. Second, we
propose a {\it robust} version of the FHSAP problem. The robust
problem aims to minimize the worst-case cost when the demand is only
known to be in a certain set. We propose an algorithm to solve the
robust FHSAP problem and show that indeed the robust formulation can
guard against data uncertainty with relatively little cost. We
believe that both of our contributions may help people to find the
desired model/approach when facing such problems in practice.

\bibliographystyle{plain}
\bibliography{hub}
\end{document}